\documentclass[12pt]{llncs}

\usepackage[utf8]{inputenc} 
\usepackage{amsmath}
\usepackage{amsfonts}
\usepackage{amssymb}
\usepackage{hyperref}

\usepackage{tikz}
\usetikzlibrary{arrows.meta}

\def\TWOOPT{2-Opt~Heuristic}
\def\2opt{2-Opt heuristic}
\def\mid{:}
\DeclareMathOperator{\area}{area}

\title{\hspace*{-7mm}\mbox{The Approximation Ratio of the \TWOOPT} for the Metric Traveling Salesman Problem}
\author{Stefan Hougardy~~ ~~Fabian Zaiser~~ ~~Xianghui Zhong}
\institute{Research Institute for Discrete Mathematics,
           University of Bonn\\
           Lenn\'estr.~2, 53113 Bonn, Germany\\[5mm]
           }
\begin{document}
\maketitle 

\centerline{\today}

\begin{abstract}
The \2opt is one of the simplest algorithms 
for finding good solutions to the metric Traveling Salesman Problem.
It is the key ingredient to the well-known Lin-Kernighan algorithm and often used in practice.
So far, only upper and lower bounds on the approximation ratio of the \2opt for the metric TSP were known. 
We prove that for the metric TSP with $n$ cities, the approximation ratio of the \2opt is $\sqrt{n/2}$ and that this bound is tight.
\end{abstract}

{\small\textbf{keywords:} traveling salesman problem; metric TSP; 2-Opt; approximation algorithm}

\section{Introduction}

In the Traveling Salesman Problem (TSP), we are given $n$ cities with their pairwise 
distances. The task is to find a shortest tour that visits each city exactly once. 
The Traveling Salesman Problem is one of the most intensely studied problems in combinatorial optimization.
It is well known to be NP-hard~\cite{GJ1979}. Without any additional assumptions, the Traveling Salesman Problem
is also hard to approximate to any number that is polynomial in $n$~\cite{SG1976}. The \emph{metric TSP}
is a special case of the TSP where the distance function satisfies the triangle inequality.
The metric TSP is also NP-hard~\cite{Kar1972}. Therefore, a lot of time has been spent to find polynomial time 
algorithms with a small approximation ratio for the metric TSP.  
In 1976, Christofides~\cite{Chr1976} proposed an algorithm 
for the metric TSP with an approximation ratio of $3/2$. To date, no polynomial time algorithm with smaller approximation ratio is known. 

For real-world instances appearing in practice, it turns out that many simple algorithms often find better solutions than 
Christofides' algorithm (see e.g.~\cite{Joh1990,Ben1992,Rei1994}). 
One of these algorithms is the \emph{\2opt}, which is the key ingredient to the well-known Lin-Kernighan algorithm~\cite{LK1973}.
Starting with an arbitrary tour, the \2opt repeatedly replaces two edges of the tour  
by two other edges, as long as this yields a shorter tour.
The \2opt stops when no further improvement can be made this way. 
A tour that the \2opt cannot improve is called \emph{2-optimal}.

Experiments on real-world instances have shown that the \2opt applied to a greedy tour achieves much better results
than Christofides' algorithm (see e.g. Bentley~\cite{Ben1992}). The exact approximation ratio
of the \2opt for metric TSP was not known so far. 
In 1987, Plesn\'\i k~\cite{Ple1987a} proved a lower bound of $\sqrt{n/8}$. 
In 1999, Chandra, Karloff, and Tovey~\cite{CKT1999} presented a proof showing an upper bound of $4\sqrt n$. 
In 2013, Levin and Yovel~\cite{LY2013} observed that this proof yields the value $2 \sqrt{2n}$. 
This leaves a gap of factor $8$ between the upper bound  $2 \sqrt{2n}$ and the lower bound $\sqrt{n/8}$.
Our main result determines the exact approximation ratio of the \2opt:

\begin{theorem}
The length of a 2-optimal tour in a metric TSP instance with $n$ cities is at most
$\sqrt{n/2}$ times the length of a shortest tour and this bound is tight.
\label{thm:main}
\end{theorem}

As the \2opt always returns a 2-optimal tour and the \2opt may start with any tour, we immediately get:

\begin{corollary}
The \2opt for metric TSP instances with $n$ cities has approximation ratio $\sqrt{n/2}$ and this result is tight.
\end{corollary}

To prove Theorem~\ref{thm:main}, we show in Section~\ref{sec:upperBound} that the length of a 2-optimal tour in a metric TSP instance 
is bounded by  $\sqrt{n/2}$ times the length of a shortest tour. In Section~\ref{sec:lowerBound}, we provide an infinite family of
metric TSP instances and 2-optimal tours within these instances with length $\sqrt{n/2}$ times the length of a shortest tour.
This proves the tightness stated in Theorem~\ref{thm:main}. Before proving the upper and the lower bound, we present in Section~\ref{sec:notation}
some notation and background on the \2opt.

\section{Metric TSP and the \TWOOPT}
\label{sec:notation}

Let $G=(V(G),E(G))$ be a complete undirected graph with $|V(G)|=n$. 
The set $E(G)$ contains all $\binom{n}{2}$ possible edges between the $n$ vertices.
The \emph{distances} between the vertices are defined by a function  $c:E(G) \to \mathbb{R}_{\ge 0}$.
A \emph{tour} in $G$ is a cycle that contains all the vertices of $G$.
The \emph{length} of a tour $T$ in $G$ is defined as $c(T) := \sum_{e\in E(T)} c(e)$.  
A \emph{shortest tour} is a tour of minimum length among the tours in $G$. 
Given a graph $G=(V(G),E(G))$ and  a function $c:E(G) \to \mathbb{R}_{\ge 0}$, the Traveling Salesman Problem
is to find a shortest tour in $G$. To simplify the notation, we will denote the length of an edge $\{x,y\} \in E(G)$ simply by $c(x,y)$ instead 
of the more cumbersome notation $c(\{x,y\})$. 
In the metric TSP, the distance function $c$ satisfies the triangle inequality, i.e.\ we have 
for any set of three vertices $x, y, z\in V(G)$:
\begin{equation}
c(x,y) ~+~ c(y, z) ~ ~\ge~ ~ c(x, z).
\end{equation}

An algorithm for the traveling salesman problem has \emph{approximation ratio}
$\alpha(n)\ge 1	$ if for every TSP instance with $n$ vertices, it finds a tour that is at most 
$\alpha(n)$ times as long as a shortest tour.

The \2opt repeatedly replaces two edges from the tour by two other edges such that the resulting tour is shorter.
Given a tour $T$ and two edges $\{a,b\}$ and $\{x,y\}$ in $T$, there are two possibilities to replace these two edges by two other edges.
Either we can choose the pair $\{a,x\}$ and $\{b,y\}$ or we can choose the pair $\{a,y\}$ and $\{b,x\}$. Exactly one of these two pairs
will result in a tour again. Without knowing the other edges of $T$, we cannot decide which of the two possibilities we have to choose.
Therefore, we will assume in the following that the tour $T$ is an \emph{oriented} cycle, i.e.\ the edges of $T$ have an orientation 
such that each vertex has exactly one incoming and one outgoing edge. Using this convention, there is only one possibility to
exchange a pair of edges such that the new edge set is a tour again: two directed edges $(a,b)$ and $(x,y)$ have to be replaced by
the edges $(a,x)$ and $(b,y)$. Note that to obtain an oriented cycle again, one has to reverse the direction of the segment between $b$ and $x$, 
see Figure~\ref{fig:2-Opt}.

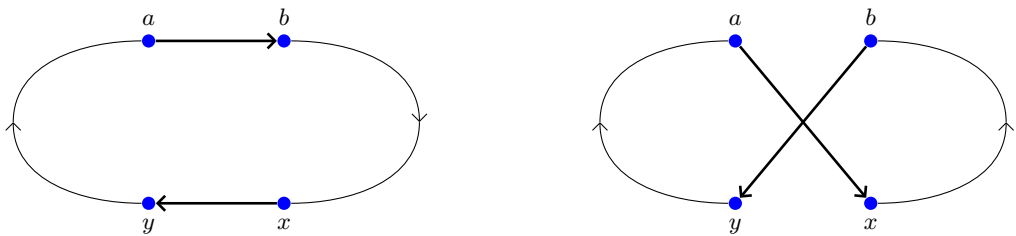
\begin{figure}
\centering
\begin{tikzpicture}[scale=1.2]
\tikzstyle{vertex}=[blue,circle,fill, minimum size=5, inner sep=0]
\tikzstyle{arrow}=[Straight Barb[length=1mm]]
\node[vertex, label=above:$a$] (a)  at (2  , 2.8) {};
\node[vertex, label=below:$y$] (y)  at (2  , 1  ) {};
\node[vertex, label=above:$b$] (b)  at (3.5, 2.8) {};
\node[vertex, label=below:$x$] (x)  at (3.5, 1  ) {};

\draw[-{Straight Barb[length=1mm]}, line width = 0.4, out =   0, in =  90] (b) to  (5  , 1.9);
\draw[                              line width = 0.4, out =   0, in = -90] (x) to  (5  , 1.9);
\draw[                              line width = 0.4, out = 180, in =  90] (a) to  (0.5, 1.9);
\draw[-{Straight Barb[length=1mm]}, line width = 0.4, out = 180, in = -90] (y) to  (0.5, 1.9);

\draw[-{Straight Barb[length=1mm]},  line width=1] (a) to (b);
\draw[-{Straight Barb[length=1mm]},  line width=1] (x) to (y);

\begin{scope}[shift={(6.5,0)}]
\node[vertex, label=above:$a$] (a)  at (2  , 2.8) {};
\node[vertex, label=below:$y$] (y)  at (2  , 1  ) {};
\node[vertex, label=above:$b$] (b)  at (3.5, 2.8) {};
\node[vertex, label=below:$x$] (x)  at (3.5, 1  ) {};

\draw[                              line width = 0.4, out =   0, in =  90] (b) to  (5  , 1.9);
\draw[-{Straight Barb[length=1mm]}, line width = 0.4, out =   0, in = -90] (x) to  (5  , 1.9);
\draw[                              line width = 0.4, out = 180, in =  90] (a) to  (0.5, 1.9);
\draw[-{Straight Barb[length=1mm]}, line width = 0.4, out = 180, in = -90] (y) to  (0.5, 1.9);

\draw[-{Straight Barb[length=1mm]},  line width=1] (a) to (x);
\draw[-{Straight Barb[length=1mm]},  line width=1] (b) to (y);

\end{scope}

\end{tikzpicture}
\caption{A TSP tour (left) and the tour obtained after replacing the edges  $(a,b)$ and $(x,y)$ with the edges $(a,x)$ and $(b,y)$ (right).
The orientation of the tour segment between the vertices $b$ and $x$ has been reversed in the new tour.}
\label{fig:2-Opt}
\end{figure}

If $(a,b)$ and $(x,y)$ are two edges in a tour $T$ and we have 
$$ c(a,x) + c(b,y) ~<~ c(a,b) + c(x,y) $$
then we say that replacing the edges $(a,b)$ and $(x,y)$ in $T$ by  the edges $(a,x)$ and $(b,y)$
is an \emph{improving 2-change}. Thus, the \2opt can be formulated as follows: \medskip

\framebox{\parbox{10cm}{
\noindent
{\bfseries \2opt} $(G=(V(G),E(G)), c:E(G) \to \mathbb{R}_{\ge 0})$\\[2mm]
1~ start with an arbitrary tour $T$\\
2~ \texttt{while} $\exists$ improving 2-change in $T$ \\
3~ ~~~~  perform an improving 2-change\\
4~ \texttt{output} $T$}}\medskip

\section{The Upper Bound on the Approximation Ratio}
\label{sec:upperBound}

Chandra, Karloff, and Tovey~\cite{CKT1999} proved in 1999 that the \2opt has an approximation ratio of $4\sqrt n$ for metric TSP. 
In 2013, Levin and Yovel~\cite{LY2013} observed that their proof yields the upper bound $2 \sqrt{2n}$. 
Here we present a new proof which improves this bound by a factor of~4:

\begin{theorem}
The approximation ratio of the \2opt on metric TSP is at most
$\sqrt{\frac{n}{2}}$.
\end{theorem}

\begin{proof}
Let $G=(V(G),E(G))$ with $c:E(G)\to\mathbb{R}_{\ge 0}$ and $|V(G)| = n$ be a metric TSP instance and let $T$ be an optimal tour. 
We may assume that $T$ has length~1. We fix an orientation of the tour $T$ and choose two vertices $p,q\in V(G)$ arbitrarily.
For each vertex $v\in V(G)$, let $i_{p}(v)$ be the length taken $\bmod$~1 of the unique shortest directed $p$-$v$ path starting in $p$ 
and using only edges of $T$. By our assumption, we have $i_{p}:V(G)\to [0,1)$ and we define $i_{q}$ similarly.
For the following, it helps to think of $[0,1)$ as the circle with circumference 1 and of $i_{p}$ as an embedding of the optimal tour into this circle such that the arc distance of two consecutive vertices on the circle is the length of the edge between them.

Define the following metric $d$ on the interval $[0,1)$, interpreted as a circle:
$d(x,y)$ is the length of the shorter of the two arcs between $x$ and $y$ on the circle, i.e., $d(x,y):=\min\{\lvert x-y \rvert , 1-\lvert x-y\rvert\}$. For any points $x,y,z \in [0,1)$ we have $d(x,y)+d(y,z)\geq d(x,z)$ since combining the two shortest arcs between $x,y$ and $y,z$ and deleting the overlap results in an arc between $x,z$.

Let $T'$ be a 2-optimal tour. As usual, we assume that it is directed.  Now consider for each edge $(u,v)$ of $T'$ the set
\begin{align*}
S_{p,q}(u,v)=\{(x,y)\in [0,1) \times [0,1) \mid d(x,i_{p}(u))+d(y,i_{q}(v))<c(u,v)\},
\end{align*}
as shown in Figure~\ref{fig:upperbound}. We claim that all these sets are pairwise disjoint for distinct edges $(u_1, v_1), (u_2, v_2) \in E(T')$.
Suppose that $S_{p,q}(u_1,v_1)$ and $S_{p,q}(u_2,v_2)$ intersect in $(x,y)$. Then, by the triangle inequality for $c$ and $d$, we have
\begin{align*}
c(u_1,u_2)+c(v_1,v_2)&\leq d(i_{p}(u_1),i_{p}(u_2))+d(i_{q}(v_1),i_{q}(v_2))\\
&\leq d(i_{p}(u_1),x)+d(x,i_{p}(u_2))+d(i_{q}(v_1),y)+ d(y,i_{q}(v_2))\\
&< c(u_1,v_1)+c(u_2,v_2).
\end{align*}
This contradicts the 2-optimality of $T'$. Hence, all these sets $S_{p,q}(u,v)$ are disjoint.

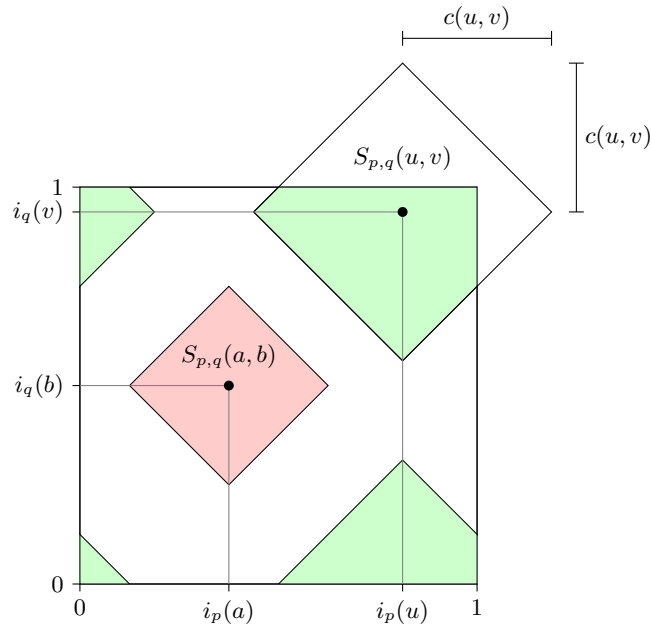
\begin{figure}
\centering
\begin{tikzpicture}[scale=0.33]
\def\firstrect{(13,9) -- (7, 15) -- (13, 21) -- (19, 15) -- (13,9)}
\def\secondrect{(6,4) -- (2, 8) -- (6, 12) -- (10, 8) -- (6,4)}
\begin{scope}
\clip[draw] (0,0) rectangle (16,16);
\begin{scope}[shift={(  0,  0)}] \filldraw[fill=green!20] \firstrect; \end{scope}
\begin{scope}[shift={(-16,  0)}] \filldraw[fill=green!20] \firstrect; \end{scope}
\begin{scope}[shift={(-16,-16)}] \filldraw[fill=green!20] \firstrect; \end{scope}
\begin{scope}[shift={(  0,-16)}] \filldraw[fill=green!20] \firstrect; \end{scope}
\end{scope}
\draw \firstrect;
\filldraw[fill=red!20] \secondrect; 
\draw (0,0) rectangle (16,16);
\draw (0,0) -- (0,-0.3) node[below] {$0$};
\draw (0,0) -- (-0.3,0) node[left]  {$0$};
\draw (16,0) -- (16,-0.3) node[below] {$1$};
\draw (0,16) -- (-0.3,16) node[left]  {$1$};

\draw (13,0) -- (13,-0.3) node[below] {$i_{p}(u)$};
\draw ( 6,0) -- ( 6,-0.3) node[below] {$i_{p}(a)$};
\draw (0,15) -- (-0.3,15) node[left] {$i_{q}(v)$};
\draw (0, 8) -- (-0.3, 8) node[left] {$i_{q}(b)$};

\draw[gray] (13,0) -- (13,15) -- (0,15);
\draw[gray] ( 6,0) -- ( 6, 8) -- (0, 8);

\fill (6,8) circle (0.2);
\fill (13,15) circle (0.2);
\draw ( 6, 10) node[below] {$S_{p,q}(a,b)$};
\draw (13, 18) node[below] {$S_{p,q}(u,v)$};

\draw (13, 22) -- (19, 22) node[midway, above] {$c(u,v)$};
\draw (13, 22.3) -- (13, 21.7);
\draw (19, 22.3) -- (19, 21.7);

\draw (20, 21) -- (20, 15) node[midway, right] {$c(u,v)$};
\draw (19.7, 21) -- (20.3, 21);
\draw (19.7, 15) -- (20.3, 15);

\end{tikzpicture}
\caption{The sets $S_{p,q}(a,b)$ (red) and $S_{p,q}(u,v)$ (green) assigned to the edges $(a,b)$ and $(u,v)$ of a 2-optimal tour. The sets are taken modulo the unit square and thus may consist of up to four parts.}
\label{fig:upperbound}
\end{figure}

Next, we want to show that the area of each set is independent of the choice of $p$ and $q$.
Let $p'$ and $q'$ be a different choice.
Note that for all vertices $u$, we have $i_{p'}(u)=i_{p'}(p) + i_{p}(u) \bmod 1$.
In particular, we find $d(x, i_p(u)) = d(x + i_{p'}(p) \bmod 1, i_{p'}(u))$ because both points are shifted by $i_{p'}(p)$ on the circle $[0,1)$.
By the definition of $S_{p,q}(u,v)$, this means that the map
\begin{align*}
t: [0,1) \times [0,1) &\to [0,1) \times [0,1) \\
(x,y) &\mapsto (x + i_{p'}(p) \bmod 1,y)
\end{align*}
bijectively sends $S_{p,q}(u,v)$ to $S_{p',q}(u,v)$.
In other words, we obtain $S_{p',q}(u,v)$ from $S_{p,q}(u,v)$ by cutting the unit square vertically at $1 - i_{p'}(p) = i_{p}(p')$ into two rectangles and reassembling them, as described by the following two translations:
\begin{align*}
t_1: [0,i_{p}(p')) \times [0,1) &\to [i_{p'}(p),1) \times [0,1) \\
(x,y) &\mapsto (x + i_{p'}(p), y) \\
t_2: [i_{p}(p'),1) \times [0,1) &\to [0,i_{p}(p')) \times [0,1) \\
(x,y) &\mapsto (x - i_{p}(p'), y)
\end{align*}
Since they have disjoint domains and disjoint images, their union $t = t_1\cup t_2$ is a bijection $[0,1) \times [0,1) \to [0,1) \times [0,1)$; sends $S_{p,q}(u,v)$ bijectively to $S_{p',q}(u,v)$; and preserves the area of this set because it consists of translations.
Analogously, we can cut the square horizontally at $i_{q}(q')$ to obtain $S_{p',q'}(u,v)$ from $S_{p',q}(u,v)$, again preserving its area.
We conclude that the area of $S_{p,q}(u,v)$ is independent of the choice of $p$ and $q$.

Now we want to show that the area of $S_{p,q}(u,v)$ is $2c(u,v)^2$ for any edge $(u,v)\in E(T')$.
By the previous paragraph, we can choose $p = u$ and $q = v$.
Then $S_{u,v}(u,v) = \{(x,y) \in [0,1) \times [0,1) \mid d(x,0) + d(y,0) < c(u,v) \}$.
This set consists of four congruent isosceles right-angled triangles whose legs have length $c(u,v)$.
Note that they do not overlap because the metric property ensures $c(u,v) \leq \frac{1}{2}$.
Hence we have: $\area(S_{p,q}(u,v)) = 4 \cdot \frac{c(u,v)^2}{2} = 2c(u,v)^2$.

Since the sets $S_{p,q}(u,v)$ for $(u,v) \in E(T')$ are pairwise disjoint, their combined area cannot exceed that of the unit square:
\[ 2\sum_{e \in E(T')} c(e)^2 = \sum_{(u,v) \in E(T')} \area(S_{p,q}(u,v)) \leq \area([0,1) \times [0,1)) = 1. \]
Then the inequality of arithmetic and quadratic means implies
$$ \frac{\sum_{e\in E(T')} c(e)}{n} \leq \sqrt{\frac{\sum_{e\in E(T')} c(e)^2}{n}} \leq \frac{1}{\sqrt{2n}}. $$
Hence, the length of the 2-optimal tour $T'$ satisfies $\sum_{e\in E(T')} c(e) \leq \sqrt{\frac{n}{2}}$.
\hfill\qed
\end{proof}

\section{The Lower Bound on the Approximation Ratio}
\label{sec:lowerBound}

To prove a lower bound $\alpha$ on the approximation ratio of the \2opt for the metric TSP, one has to 
show that for infinitely many $n$, there exists a metric TSP instance with $n$ cities
that contains a 2-optimal tour which is $\alpha$ times longer than a shortest tour.

In 1999, Chandra, Karloff, and Tovey~\cite{CKT1999} provided such a construction for all $n$ of the form $4\cdot k^2$ for positive integers $k$, which shows a lower bound of  $\frac{1}{4} \sqrt n$. Several years earlier, Plesn\'\i k~\cite{Ple1987a} had given
another construction without explicitly stating a lower bound. It turns out that his construction yields a lower bound
of $\frac{1}{\sqrt 8} \sqrt n$ and works for all $n$ of the form $8\cdot k^2 - 8\cdot k + 3$ for positive integers $k$.

The following result improves Plesn\'\i k's lower bound by a factor of 2, and yields the tight result stated in Theorem~\ref{thm:main}.

\begin{theorem}
The approximation ratio of the \2opt on the metric TSP is at least
$\sqrt{\frac{n}{2}}$.
\end{theorem}
\begin{proof}
Let $G$ be a complete graph on $n := 2\cdot k^2$ nodes with vertex set 
$V(G) := \{v_{i,j} \mid 1\le i,j \le k\} \cup \{w_{i,j} \mid 1\le i,j \le k\}$.
For each $i$ with $1\le i\le k$, we call $V_i := \{ v_{i,j} \mid 1\le j \le k \}$
and $W_i := \{ w_{i,j} \mid 1 \le j \le k \}$ a \emph{section} of $V(G)$
and the $v$-vertices and $w$-vertices the two \emph{halves} of $V(G)$.

We define a distance function $c:E(G)\to\mathbb{R}_{\ge 0}$ as follows:
\begin{align*}
c(v_{i,j}, w_{i',j'}) & =  ~~1  && \mbox{ for all } 1 \le i, i', j, j' \le k \\
c(v_{i,j}, v_{i',j'}) & = 
  \begin{cases}
  0 & i = i' \\
  2 & i \neq i'
  \end{cases} && \mbox{ for all } 1 \le j, j' \le k \\
c(w_{i,j}, w_{i',j'}) & = 
  \begin{cases}
  0 & i = i' \\
  2 & i \neq i'
  \end{cases} && \mbox{ for all } 1 \le j, j' \le k 
\end{align*}

It is not hard to see that the function $c$ satisfies the triangle inequality:
Let $u,v,w$ be any three vertices in $V(G)$.
We want to show that $c(u,w) \leq c(u,v) + c(v,w)$.
As $c$ takes only the values $0,1,2$, this is obvious if $c(u,v) \geq 1$ and $c(v,w) \geq 1$.
Otherwise, without loss of generality, we may assume $c(u,v) = 0$.
i.e., $u$ and $v$ are in the same section of $V(G)$.
But then the definition of $c$ implies $c(u,w) = c(v,w)$ and the triangle inequality is satisfied.
Therefore, the graph $G$ with cost function $c$ is a metric TSP instance.

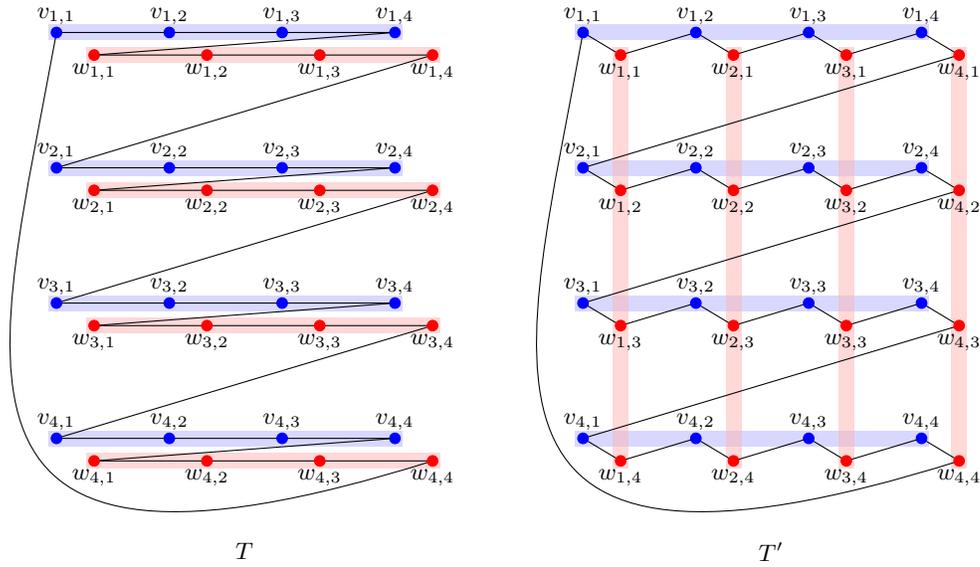
\begin{figure}
\centering

\hspace*{-12mm}\begin{tikzpicture}
\tikzstyle{vertex_v}=[blue,circle,fill,draw, minimum size=4, outer sep = -0.7mm, inner sep=0]
\tikzstyle{vertex_w}=[ red,circle,fill,draw, minimum size=4, outer sep = -0.7mm, inner sep=0]
\def\yscale{1.8}
\begin{scope}[shift={(0,0)}]
\foreach \y in {0,1,2,3}
    {\fill[blue!30,opacity=0.5] (0-0.1,\yscale*\y - 0.1) rectangle (3 * 1.5 + 0.1,\yscale*\y + 0.1);
     \fill[red!30 ,opacity=0.5] (0-0.1 + 0.5,\yscale*\y - 0.1 - 0.3) rectangle (3 * 1.5 + 0.1 + 0.5,\yscale*\y + 0.1 - 0.3);}

\foreach \x in {0,1,2}
    \foreach \y in {0,1,2,3} 
        {\draw[] (1.5*\x,  \yscale*\y) -- (1.5*\x + 1.5,  \yscale*\y);
         \draw[] (0.5 + 1.5*\x,  -0.3 + \yscale*\y) -- (0.5 + 1.5*\x + 1.5,  -0.3 + \yscale*\y);}
         
\foreach \y in {0,1,2,3} 
    \draw[] (1.5*3,  \yscale*\y) -- (0.5 + 1.5*0,  -0.3 + \yscale*\y);

\foreach \y in {1,2,3} 
    \draw[] (0.5 + 1.5*3,  -0.3 + \yscale*\y) -- (1.5*0,  \yscale*\y - \yscale);
    
\draw[] (0.5 + 1.5 * 3, -0.3 + \yscale * 0) .. controls (-2, -2.5)  and (-0.8, 1) .. (1.5 * 0, \yscale * 3);
         
\foreach \x in {1,2,3,4}
	\foreach \y in {1,2,3,4}
        {\node[vertex_v, label=above:$v_{\y,\x}$] (v\x,\y) at (1.5*\x - 1.5,        4 * \yscale - \yscale*\y) {};
         \node[vertex_w, label=below:$w_{\y,\x}$] (w\y,\x) at (1.5*\x - 1.0, -0.3 + 4 * \yscale - \yscale*\y) {};}
         
\draw (2.5,-1.5) node {$T$};         
\end{scope}

\begin{scope}[shift={(7,0)}]
\foreach \y in {0,1,2,3}
    \fill[blue!30,opacity=0.5] (0-0.1,\yscale*\y - 0.1) rectangle (3 * 1.5 + 0.1,\yscale*\y + 0.1);

\foreach \x in {0,1,2,3}
    \fill[red!30 ,opacity=0.5] (0-0.1 + 1.5*\x + 0.5, - 0.1 - 0.3) rectangle (\x * 1.5 + 0.1 + 0.5, 3*\yscale + 0.1 - 0.3);

\foreach \x in {0,1,2,3}
    \foreach \y in {0,1,2,3} 
        \draw[] (1.5*\x,  \yscale*\y) -- (0.5 + 1.5*\x,  -0.3 + \yscale*\y);

\foreach \x in {0,1,2}
    \foreach \y in {0,1,2,3} 
        \draw[] (1.5*\x + 1.5,  \yscale*\y) -- (0.5 + 1.5*\x,  -0.3 + \yscale*\y);

\foreach \y in {1,2,3} 
    \draw[] (0.5 + 1.5*3,  -0.3 + \yscale*\y) -- (1.5*0,  \yscale*\y - \yscale);
    
\draw[] (0.5 + 1.5 * 3, -0.3 + \yscale * 0) .. controls (-2, -2.5)  and (-0.8, 1) .. (1.5 * 0, \yscale * 3);
         
\foreach \x in {1,2,3,4}
	\foreach \y in {1,2,3,4}
        {\node[vertex_v, label=above:$v_{\y,\x}$] (v\x,\y) at (1.5*\x - 1.5,        4 * \yscale - \yscale * \y) {};
         \node[vertex_w, label=below:$w_{\x,\y}$] (w\x,\y) at (1.5*\x - 1.0, -0.3 + 4 * \yscale - \yscale * \y) {};}

\draw (2.5,-1.5) node {$T'$};         
\end{scope}        
\end{tikzpicture}\\[-1cm]
\caption{The optimal tour $T$ (left) and the 2-optimal tour $T'$ (right) for $k = 4$.
Note that the $w$-vertices on the right are mirrored at the diagonal compared
to the $w$-vertices on the left. Thus, on the left, vertices within the sections $V_i$ and $W_i$ are in a row. On the right, the vertices in 
the sections $V_i$ are in a row while the vertices in a section $W_i$ are within a column. 
The colored bars contain the vertices belonging to the same section.}
\label{metric-picture}
\end{figure}

In the following, we will construct two special tours in $G$,
which are depicted in Figure~\ref{metric-picture}.
Let $T$ be the tour consisting of the edges
\begin{align*}
E(T) =~ & \{(v_{i,j}, v_{i,j+1}) \mid 1 \le i \le k, 1 \le j < k \} ~\cup \\
        & \{(w_{i,j}, w_{i,j+1}) \mid 1 \le i \le k, 1 \le j < k \} ~\cup \\
        & \{(v_{i,k}, w_{i,1})   \mid 1 \le i \le k \} ~\cup \\
        & \{(w_{i,k}, v_{i+1,1}) \mid 1 \le i < k \} ~\cup \\
        & \{(w_{k,k}, v_{1,1})\}.
\end{align*}
The edges in the first two  sets have length 0; the $2k$ edges in the other three sets have length 1.
Therefore, we have $c(T) = 2k$.
This tour is optimal because any tour has to visit all $2k$ sections of $V(G)$
and the distance of two vertices from different sections is at least 1.

Next we consider the tour $T'$ with
\begin{align*}
E(T') =~ & \{(v_{i,j}, w_{j,i})   \mid 1\le i, j \le k \} ~\cup \\
         & \{(w_{j,i}, v_{i,j+1}) \mid 1\le i \le k, 1 \le j < k\} ~\cup \\
         & \{(w_{k,i}, v_{i+1,1}) \mid 1 \le i < k \} ~\cup \\
         & \{(w_{k,k}, v_{1,1})\}.
\end{align*}
Each edge of $T'$ has length 1.
Thus we have $c(T') = 2k^2$.
We claim that the tour $T'$ is 2-optimal.
Assume by contradiction that $T'$ is not 2-optimal.
Consider a pair of edges $(a,b),(x,y)$
that allows an improving 2-change to $(a,x),(b,y)$.
Hence $c(a,x) + c(b,y) < c(a,b) + c(x,y) = 2$
and one of $c(a,x)$ or $c(b,y)$ must be zero.
This means $a$ and $x$ or $b$ and $y$ must be in the same section.
But since $a$ and $b$ are in opposite halves of $V(G)$ (just like $x$ and $y$),
this means that $a$ and $x$ are in one half of $V(G)$ and $b$ and $y$ in the other.
Hence $c(a,x),c(b,y) \in \{0,2\}$.
For an improving 2-change, we must have $c(a,x) = c(b,y) = 0$. 
This implies that $a$ and $x$ lie in the same section of $V(G)$ and $b$ and $y$ lie in the same section of $V(G)$.
Thus there must exist indices $i$ and $j$ with $1 \le i,j \le k$ such that $a,x \in V_i$ and $b, y \in W_j$ or such that $a,x \in W_i$
and $b, y \in V_j$. This implies that there must exist two different edges from $V_i$ to $W_j$ or from $W_i$ to $V_j$.
However, this is a contradiction as by definition of $T'$, for any pair $i,j$ with $1\le i,j\le k$, there exists exactly one edge 
directed from $V_i$ to $W_j$ (namely the edge $(v_{i,j}, w_{j,i})$) and exactly one edge directed from $W_j$ to $V_i$. 
This proves the 2-optimality of $T'$.

Combining the above findings we get
\[ \frac{c(T')}{c(T)} = \frac{2k^2}{2k}
    = k = \sqrt{\frac{2k^2}{2}} = \sqrt{\frac{n}{2}}. \]
    \hfill\qed
\end{proof}

\bibliographystyle{plain}

\end{document}